\documentclass[11pt]{article}
\usepackage{amsthm}
\usepackage{graphicx}

\graphicspath{{./graphics/}}

\newtheorem{theorem}{Theorem}
\newtheorem{lemma}[theorem]{Lemma}
\newtheorem{corollary}[theorem]{Corollary}

\begin{document}

\title{Some Results on the Circuit Complexity of\\ Bounded Width Circuits and Nondeterministic Circuits}

\author{Hiroki Morizumi\\
{\small Shimane University, Japan}\\
{\small morizumi@cis.shimane-u.ac.jp}
}

\date{}

\maketitle

\begin{abstract}
In this paper, we consider bounded width circuits and nondeterministic circuits
in three somewhat new directions.
In the first part of this paper, we mainly consider bounded width circuits.
The main purpose of this part is to prove that there is a Boolean function $f$
which cannot be computed by any nondeterministic circuit of size $O(n)$
and width $o(n)$.  
To the best of our knowledge, this is the first result on the lower bound
of (nonuniform) bounded width circuits computing an explicit Boolean function,
even for deterministic circuits.
Actually, we prove a more generalized lower bound.
Our proof outline for the lower bound also provides a satisfiability algorithm
for nondeterministic bounded width circuits.
In the second part of this paper, we consider the power of nondeterministic
circuits. We prove that there is a Boolean function $f$ such that
the nondeterministic $U_2$-circuit complexity of $f$ is at most $2n + o(n)$
and the deterministic $U_2$-circuit complexity of $f$ is $3n - o(n)$.
This is the first separation on the power of deterministic and
nondeterministic circuits for general circuits.
In the third part of this paper, we show a relation between deterministic
bounded width circuits and nondeterministic bounded width circuits.
As the main consequence, we prove that
$\mathsf{L/quasipoly} \supseteq \mathsf{NL/poly}$.
As a corollary, we obtain that $\mathsf{L/quasipoly} \supset \mathsf{NL}$.
To the best of our knowledge, this is the first result on $\mathsf{L}$
with large (more precisely, superpolynomial size) advice.
\end{abstract}

\section{Introduction}

Bounded width circuits and nondeterministic circuits are computation models
related to bounded space computation and nondeterministic computation.
In this paper, we consider bounded width circuits and nondeterministic circuits
in three somewhat new directions.
We describe the three directions in Section~\ref{subsec:intro_bwc},
Section~\ref{subsec:intro_ndc} and Section~\ref{subsec:intro_l}, respectively.
Section~\ref{subsec:intro_bwc}, Section~\ref{subsec:intro_ndc} and
Section~\ref{subsec:intro_l} correspond to Section~\ref{sec:bwc},
Section~\ref{sec:ndc} and Section~\ref{sec:l}, respectively.

\subsection{Bounded width circuits} \label{subsec:intro_bwc}

We prove the following theorem and corollary.
(Note our definition of nondeterministic bounded width circuits.
 See Section~\ref{sec:pre}.)
\begin{theorem} \label{thrm:bwc_main}
There is a Boolean function $f$ as follows: If a nondeterministic circuit
of size $s$ and width $w$ computes $f$, then
$$w = \Omega(\frac{n^4}{4^{\frac{s}{n}}s^3}) - \frac{\log_2 s}{2}.$$
\end{theorem}
\begin{corollary} \label{coro:bwc_main}
There is a Boolean function $f$ which cannot be computed by any
nondeterministic circuit of size $O(n)$ and width $o(n)$.  
\end{corollary}

To the best of our knowledge, this is the first result on the lower bound
of (nonuniform) bounded width circuits computing an explicit Boolean function,
even for deterministic circuits.
Although Theorem~\ref{thrm:bwc_main} belongs to the study of time-space
tradeoffs, our computation models are completely nonuniform.

Proving that there is a Boolean function $f$ which cannot be computed by any
(deterministic) circuit of size $O(n)$ and depth $O(\log n)$ is one of
realistic goals in circuit complexity.
Corollary~\ref{coro:bwc_main} resolves the width variant of the open problem
in a stronger form.
Proving a superlinear size lower bound for general circuits is a central
problem in circuit complexity.
Corollary~\ref{coro:bwc_main} implies that we can prove a superlinear
lower bound if the width is slightly bounded.

The proof of our lower bounds is based on the size lower bound for
nondeterministic syntactic read-$k$-times branching programs
in 1993~\cite{BRS93}.
The relation has not been known for a long time.

Our proof outline for the lower bound also provides a satisfiability algorithm
for nondeterministic bounded width circuits.
See Section~\ref{subsec:bwc_sat} for the details.

\subsection{The power of nondeterministic circuits} \label{subsec:intro_ndc}

Nondeterministic circuits are a nondeterministic variant of Boolean circuits
as a computation model.
While both of nondeterministic computation and circuit complexity
are central topics in computational complexity, the circuit complexity
of nondeterministic circuits is relatively not well studied.
The author proved a $3(n-1)$ lower bound for the size of nondeterministic
$U_2$-circuits computing the parity function in his previous paper~\cite{M15}.
It was known that the minimum size of deterministic $U_2$-circuits computing
the parity function exactly equals $3(n-1)$~\cite{S74}.
Thus, nondeterministic computation is useless to compute the parity function by $U_2$-circuits.

In this paper, we consider the opposite directions, i.e., the case
that nondeterministic computation is useful.
We denote by $size^{\rm dc}(f)$ the size of the smallest deterministic
$U_2$-circuit computing a function $f$, and denote by $size^{\rm ndc}(f)$
the size of the smallest nondeterministic $U_2$-circuit computing
a function $f$.
We prove the following theorem.

\begin{theorem} \label{thrm:ndc_main}
There is a Boolean function $f$ such that $size^{\rm ndc}(f) \leq 2n + o(n)$
and $size^{\rm dc}(f) = 3n - o(n)$.  
\end{theorem}

To prove Theorem~\ref{thrm:ndc_main}, we introduce a simple proof
strategy, and call the key idea {\em nondeterministic selecting}.
In Section~\ref{subsec:ndsel}, we explain nondeterministic selecting
and the proof outline using it.

\subsection{Bounded space, nondeterminism, and large advice} \label{subsec:intro_l}

It is easily confirmed that deterministic bounded width circuits
with exponential size can compute an arbitrary Boolean function
even if the width is three.
In this third part, we consider deterministic bounded width circuits
with large size, which also means that we consider bounded space computation
with large advice in Turing machines.
We especially consider deterministic circuits of width $O(\log n)$
and quasipolynomial size.

Actually, we prove the following relation between deterministic
bounded width circuits and nondeterministic bounded width circuits.
\begin{theorem} \label{thrm:l_conv}
Any nondeterministic circuit of size $s$ and width $w$ can be converted
to a deterministic circuit of size $2^{O((w + \log s)\log s)}$ and width $w + O(\log s)$.
\end{theorem}
As the main consequence, we prove the following theorem.
\begin{theorem} \label{thrm:l_main}
$\mathsf{L/quasipoly} \supseteq \mathsf{NL/poly}$.
\end{theorem}
See Section~\ref{subsec:l_pre} for the definitions.
Since $\mathsf{NL/poly} \supset \mathsf{NL}$, the following corollary
is immediately obtained.
\begin{corollary} \label{coro:l_main}
$\mathsf{L/quasipoly} \supset \mathsf{NL}$.
\end{corollary}

We consider Theorem~\ref{thrm:l_main} and Corollary~\ref{coro:l_main}
from three points of view below.
\\

\noindent {\bf The {\sf L} vs. {\sf NL} problem.}
Savitch's theorem~\cite{S70} shows that
$\mathrm{NSPACE}(f(n)) \subseteq \mathrm{SPACE}(f(n)^2)$ for $f(n) \geq \log n$.
While $\mathsf{PSPACE} = \mathsf{NPSPACE}$ by the theorem,
the {\sf L} vs. {\sf NL} problem is a longstanding central open problem
in computational complexity.
Corollary~\ref{coro:l_main} may give some new insight for the
{\sf L} vs. {\sf NL} problem.
Savitch's theorem means that nondeterministic computation can be
replaced by more spaces in this situation.
Corollary~\ref{coro:l_main} means that nondeterministic computation can be
replaced by advice in the situation.
\\

\noindent {\bf The {\sf L/poly} vs. {\sf NL/poly} problem.}
This is the nonuniform variant of the {\sf L} vs. {\sf NL} problem
and also a longstanding open problem in computational complexity.
Theorem~\ref{thrm:l_main} can be considered as a result related to
the {\sf L/poly} vs. {\sf NL/poly} problem.
\\

\noindent {\bf The power of large advice.}
If we consider nonuniform variant of {\sf L}, then the size of
advice is polynomial.
Therefore, {\sf L/poly} has been well studied.
Theorem~\ref{thrm:l_main} and Corollary~\ref{coro:l_main} imply
nontrivial results on the power of large advice.
\\

In Section~\ref{subsec:l_conv_p} and Section~\ref{subsec:l_main_p},
we prove Theorem~\ref{thrm:l_conv} and Theorem~\ref{thrm:l_main},
respectively.
\\

\noindent {\bf Note.}
If we wish to prove only Theorem~\ref{thrm:l_main}, then we can use
nondeterministic branching programs instead of nondeterministic
bounded width circuits.
Theorem~\ref{thrm:l_conv} is replaced by the following theorem,
and the proof of Theorem~\ref{thrm:l_conv2} is almost the same
of Theorem~\ref{thrm:l_conv}.
\begin{theorem} \label{thrm:l_conv2}
Any nondeterministic branching programs of size $s$ can be converted
to a Boolean circuit of size $2^{O(\log^2 s)}$ and width $O(\log s)$.
\end{theorem}

\section{Preliminaries} \label{sec:pre}

The definitions in this section are used throughout this paper.

{\em Circuits} are formally defined as directed acyclic graphs.
The nodes of in-degree 0 are called {\em inputs}, and each one of them
is labeled by a variable or by a constant 0 or 1.
The other nodes are called {\em gates}, and each one of them
is labeled by a Boolean function.
The {\em fan-in} of a node is the in-degree of the node, and
the {\em fan-out} of a node is the out-degree of the node.
There is a single specific node called {\em output}.
The {\em size} of a circuit is the number of gates in the circuit.

While the gate type is critical in Section~\ref{sec:ndc}, it is not
so critical in Section~\ref{sec:bwc} and Section~\ref{sec:l}.
In Section~\ref{sec:bwc} and Section~\ref{sec:l}, we assume that
the gates are AND gates of fan-in two, OR gates of fan-in two, and NOT gates.
In Section~\ref{sec:ndc}, we use gates as follows.
We denote by $B_2$ the set of all Boolean functions
$f:\{0,1\}^2 \rightarrow \{0,1\}$.
We denote by $U_2$ the set of all Boolean functions over two variables
except for the XOR function and its complement.
A Boolean function in $U_2$ can be represented as the following form:
$$f(x,y) = ((x \oplus a) \wedge (y \oplus b)) \oplus c,$$
where $a, b, c \in \{0,1\}$.
A {\em $U_2$-circuit} is a circuit in which each gate has fan-in 2
and is labeled by a Boolean function in $U_2$.

When we consider the {\em width} of a circuit, we temporarily
insert COPY gates to the circuit.
A COPY gate is a dummy gate which simply outputs its input.
A circuit is {\em layered} if its set of gates can be partitioned
into subsets called {\em layers} such that every edge in the circuit
is between adjacent layers.
Note that every circuit is naturally converted to a layered circuit
by inserting COPY gates to each edge which jumps over some layers.
The width of a layer is the number of gates in the layer.
The width of a circuit is the maximum width of all layers in
the circuit.

A {\em nondeterministic circuit} is a circuit with {\em actual inputs}
$(x_1, \ldots, x_n) \in \{0,1\}^n$ and some further inputs
$(y_1, \ldots, y_m) \in \{0,1\}^m$ called {\em guess inputs}.
A nondeterministic circuit computes a Boolean function $f$ as follows:
For $x \in \{0,1\}^n$, $f(x)=1$ iff there exists a setting of the guess inputs
$\{y_1, \ldots, y_m\}$ which makes the circuit output 1.
We call a circuit without guess inputs
a {\em deterministic circuit} to distinguish it
from a nondeterministic circuit.

To the best of our knowledge, this is the first paper which consider
nondeterministic bounded width circuits, and we need to notice that
the appearance of guess inputs is a sensitive problem to the
computational power of circuits.
We restrict the number of nodes labeled by a guess input to at most one,
and we do not restrict the number of nodes labeled by an actual input.
(The restriction makes no sense if the width is unbounded.)
This is a natural restriction, since it make nondeterministic bounded width
circuits correspond to nondeterministic bounded space computation
such as {\sf NL/poly}, and we actually use this property
in Section~\ref{sec:l}.

\section{Bounded width circuits} \label{sec:bwc}

\subsection{Preliminaries}

A {\em nondeterministic branching program} is a directed acyclic graph.
The nodes of non-zero out-degree are called {\em inner nodes} and labeled
by a variable.
The nodes of out-degree 0 are called {\em sinks} and labeled by 0 or 1.
For each inner node, outgoing edges are labeled by 0 or 1.
There is a single specific node called the {\em start node}.
The output of the nondeterministic branching program is 1 if and only if
at least one path leads to 1 sink.
The {\em size} of branching programs is the number of its nodes.
A branching program is {\em syntactic read-$k$-times} if each variable
appears at most $k$ times in each path.

To prove Theorem~\ref{thrm:bwc_main}, we use the following theorem.

\begin{theorem}[\cite{BRS93}] \label{thrm:bwc_brs}
There is a Boolean function $f$ such that every nondeterministic
syntactic read-$k$-times branching program for computing $f$ has
size $\exp(\Omega(\frac{n}{4^kk^3}))$.
\end{theorem}

In this paper, we denote by $f_{BRS}$ the Boolean function of the theorem.

\subsection{Proof of Theorem~\ref{thrm:bwc_main}}

Firstly, we define the concept of {\em read-$k$-times} for circuits.
A variable is read-$k$-times in a circuit if the number of nodes labeled
by the variable is at most $k$.
A circuit is read-$k$-times if every actual input is read-$k$-times
in the circuit.
(Note that ``read-$k$-times'' makes a sense since the width is bounded.)

\begin{lemma} \label{lem:bwc_conv}
Any nondeterministic read-$k$-times circuit of size $s$ and width $w$
can be converted to a nondeterministic syntactic read-$k$-times
branching program of size $4^ws$.
\end{lemma}

\begin{proof}
The number of values (0 or 1) from a layer to the next layer is
at most $w$.
For each at most $2^w$ combination of 0 and 1, we prepare one node
in the constructed branching program.
Natural conversion from the circuit to the branching program is
enough to prove the lemma. % todo
\end{proof}

If the size of a circuit is $s$, then the number of nodes labeled by
actual inputs is at most $s+1$.
Thus, the average number of nodes labeled by each actual input is
at most $\frac{s+1}{n}$.
However, the circuit is not necessarily read-$\frac{s+}{n}$-times.
This is the most difficult point of this proof.
We resolve the difficulty by the definition of a Boolean function $f$.

We define $f(x_1, x_2, \ldots, x_{2n}, z_1, z_2, \ldots, z_{2n})$ as follows. 
If $\sum_{i=1}^{2n} z_i \neq n$, then $f = 0$.
Otherwise, $f = f_{BRS}$ and the $n$ input variables are $x_i$'s such that
$z_i = 1$.

\begin{proof}[Proof of Theorem~\ref{thrm:bwc_main}]
Let $C$ be a nondeterministic circuit computing
$f(x_1, x_2, \ldots, x_{2n}, z_1, z_2, \ldots, z_{2n})$, and let $s$ and $w$
be the size and the width of $C$, respectively.
We choose $n$ variables from $x_1, \ldots, x_{2n}$ so as every choosed
variable is read-$\frac{s}{n}$-times in $C$.
We assign $1$ to $z_i$ iff $x_i$ has been chosen for $1 \leq i \leq 2n$.
We assign an arbitrary value to $x_i$ which has not been chosen
for $1 \leq i \leq 2n$.
Let $C'$ be the obtained read-$\frac{s}{n}$-times circuit, and let $s'$ and $w'$
be the size and the width of $C'$, respectively.
$C'$ computes $f_{BRS}$, and $s' \leq s$ and $w' \leq w$.
By Lemma~\ref{lem:bwc_conv} and Theorem~\ref{thrm:bwc_brs},
\begin{eqnarray*}
  4^{w'}s'       & = & \exp(\Omega(\frac{n}{4^{\frac{s}{n}}(\frac{s}{n})^3})) \\
  4^ws          & = & \exp(\Omega(\frac{n}{4^{\frac{s}{n}}(\frac{s}{n})^3})) \\
  2w + \log_2 s & = & \Omega(\frac{n^4}{4^{\frac{s}{n}}s^3}) \\
  w             & = & \Omega(\frac{n^4}{4^{\frac{s}{n}}s^3}) - \frac{\log_2 s}{2}
\end{eqnarray*}
\end{proof}

\subsection{Satisfiability algorithms} \label{subsec:bwc_sat}

Recently, a satisfiability algorithm for nondeterministic syntactic
read-$k$-times branching programs has been provided.

\begin{theorem}[\cite{NST17}] \label{thrm:bwc_nst}
There exists a deterministic and polynomial space algorithm for
a nondeterministic and syntactic read-$k$-times BP SAT with $n$ variables
and $m$ edges that runs in time
$O(\mathrm{poly}(n, m^{k^2}) \cdot 2^{(1 - 4^{-k-1})n})$.
\end{theorem}

In a similar outline to the proof of the lower bound, a satisfiability
algorithm for nondeterministic bounded width circuits is provided.

\begin{theorem}
There exists a deterministic and polynomial space algorithm for
a nondeterministic bounded width circuit SAT with $n$ actual inputs,
size $s$ and width $w$ that runs in time
$O(\mathrm{poly}(n, m^{k^2}) \cdot 2^{(1 - 4^{-k-\frac{3}{2}})n})$,
where $m = 4^ws$ and $k = \lceil \frac{2s}{n} \rceil$. % —vŠm"F
\end{theorem}

\begin{proof}
Let $C$ be a nondeterministic bounded width circuit.
We choose $n/2$ variables from $x_1, \ldots, x_n$ so as every choosed
variable is read-$\lceil \frac{2s}{n} \rceil$-times in $C$.
For $n/2$ variables which have not been chosen, we execute
the brute-force search.
Then, we use Lemma~\ref{lem:bwc_conv}, and execute the algorithm
in Theorem~\ref{thrm:bwc_nst}.
\end{proof}

\section{The power of nondeterministic circuits} \label{sec:ndc}

\subsection{Preliminaries}

The parity function of $n$ inputs $x_1, \ldots, x_n$, denoted by Parity$_n$,
is 1 iff $\sum x_i \equiv 1 \pmod{2}$.
Circuits are $U_2$-circuits throughout Section~\ref{sec:ndc}.

\subsubsection{the gate elimination method} \label{subsubsec:method}

In our proof, we need the gate elimination method and the result
by Schnorr using the method.
In this subsection, we have a quick look at them.

Consider a gate $g$ which is labeled by a Boolean function in $U_2$.
Recall that any Boolean function in $U_2$ can be represented as
the following form:
$$f(x,y) = ((x \oplus a) \wedge (y \oplus b)) \oplus c,$$
where $a, b, c \in \{0,1\}$.
If we fix one of two inputs of $g$ so that $x = a$ or $y = b$,
then the output of $g$ becomes a constant $c$.
In such case, we call that $g$ is {\em blocked}.

\begin{theorem}[Schnorr~\cite{S74}] \label{thrm:par}
$$size^{\rm dc}({\rm Parity}_n) = 3(n-1).$$
\end{theorem}

\begin{proof}
Assume that $n \geq 2$.
Let $C$ be an optimal deterministic $U_2$-circuit computing Parity$_n$.
Let $g_1$ be a top gate in $C$, i.e.,  whose two inputs are connected from
two inputs $x_i$ and $x_j$, $1 \leq i, j \leq n$.
Then, $x_i$ must be connected to another gate $g_2$, since, if $x_i$ is
connected to only $g_1$, then we can block $g_1$ by an assignment
of a constant to $x_j$ and the output of $C$ becomes independent from $x_i$,
which contradicts that $C$ computes Parity$_n$.
By a similar reason, $g_1$ is not the output of $C$.
Let $g_3$ be a gate which is connected from $g_1$.
See Figure~\ref{fig:sch}.

\begin{figure}[t]
  \begin{center}
    \includegraphics[scale=0.6]{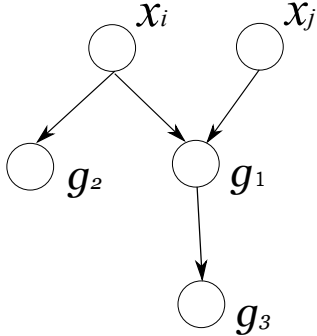}
    \caption{Proof of Theorem~\ref{thrm:par}}
    \label{fig:sch}
  \end{center}
\end{figure}

We prove that we can eliminate at least three gates from $C$ by an assignment
to $x_i$.
We assign a constant 0 or 1 to $x_i$ such that $g_1$ is blocked.
Then, we can eliminate $g_1$, $g_2$ and $g_3$.
If $g_2$ and $g_3$ are the same gate, then the output of $g_2$ ($= g_3$) 
becomes a constant, which means that $g_2$ ($= g_3$) is not the output
of $C$ and we can eliminate another gate which is connected from
$g_2$ ($= g_3$).
Thus, we can eliminate at least three gates and the circuit come to compute
Parity$_{n-1}$ or $\neg$Parity$_{n-1}$.
For deterministic circuits, it is obvious that 
$size^{\rm dc}({\rm Parity}_{n-1}) = size^{\rm dc}(\neg{\rm Parity}_{n-1})$.
Therefore,
\begin{eqnarray*}
size^{\rm dc}({\rm Parity}_n) & \geq   & size^{\rm dc}({\rm Parity}_{n-1}) + 3 \\
                            & \vdots & \\
                            & \geq   & 3(n-1).
\end{eqnarray*}

$x \oplus y$ can be computed with three gates by the following form:
$$(x \wedge \neg y) \vee (\neg x \wedge y).$$
Therefore, $size^{\rm dc}({\rm Parity}_n) \leq 3(n-1)$.
\end{proof}

\subsection{Nondeterministic selecting} \label{subsec:ndsel}

In this subsection, we describe our idea of the proof.
We call the key idea nondeterministic selecting.

Let $f':\{0,1\}^{\sqrt n} \to \{0,1\}$, and
$$f = \bigvee_{i=0}^{\sqrt n -1}f'(x_{\sqrt n \cdot i + 1},
  x_{\sqrt n \cdot i + 2}, \ldots, x_{\sqrt n \cdot i + \sqrt n}).$$
Nondeterministic circuits can compute $f$ efficiently.
We construct a nondeterministic circuit $C$ computing $f$ as follows.
Firstly, we select $\sqrt n$ inputs nondeterministically.
More precisely, we construct a selector circuit $C'$ which outputs
$x_{\sqrt n \cdot i + 1}, x_{\sqrt n \cdot i + 2}, \ldots, x_{\sqrt n \cdot i + \sqrt n}$
for each $i$, $0 \leq i \leq \sqrt n - 1$, when guess inputs of $C$
are assigned to an assignment.
Then, one circuit $C''$ computing $f'$ is enough in $C$.
$\sqrt n$ variables of the output of $C'$ are connected to the input
of $C''$.
It is not difficult to confirm that $C$ computes $f$ by the definition
of nondeterministic circuits.

On the other hand, a trivial construction of deterministic circuits
computing $f$ needs $\sqrt n$ circuits computing $f'$.
Note that it is a complicated problem (called a direct sum) whether
$\sqrt n$ circuits are needed.
In our proof of Theorem~\ref{thrm:ndc_main}, we choose the parity
function as $f'$ so that we can prove the large lower bound
of $size^{\rm dc}(f)$.

\subsection{Proof of Theorem~\ref{thrm:ndc_main}}

To prove Theorem~\ref{thrm:ndc_main}, we let
$$f = \bigvee_{i=0}^{\sqrt n -1}{\rm Parity}_{\sqrt n}(x_{\sqrt n \cdot i + 1},
  x_{\sqrt n \cdot i + 2}, \ldots, x_{\sqrt n \cdot i + \sqrt n}),$$
and prove two lemmas.

\begin{lemma} \label{lem:ndc_ub}
$size^{\rm ndc}(f) \leq 2n + o(n)$.
\end{lemma}

\begin{proof} % todo
We construct a nondeterministic circuit computing $f$ as mentioned
in Section~\ref{subsec:ndsel}.
We use $\lceil \log \sqrt n \rceil$ guess inputs.
The number of gates in the selector circuit is $2n + o(n)$.
The number of gates in one circuit computing ${\rm Parity}_{\sqrt n}$
is $o(n)$ by Theorem~\ref{thrm:par}.
\end{proof}

\begin{lemma} \label{lem:ndc_lb}
$size^{\rm dc}(f) = 3n - o(n)$.
\end{lemma}

\begin{proof}
Since $size^{\rm dc}({\rm Parity}_n) = 3(n-1)$ by Theorem~\ref{thrm:par},
$size^{\rm dc}(f) \leq 3n - o(n)$.

We prove that $size^{\rm dc}(f) \geq 3n - o(n)$. We refer the proof of
Theorem~\ref{thrm:par}.
While we eliminate at least three gates from the circuit by
an assignment to $x_i$ as the proof of Theorem~\ref{thrm:par},
we modify the proof as follows.
If $x_{\sqrt n \cdot i + 1}, x_{\sqrt n \cdot i + 2}, \ldots, x_{\sqrt n \cdot i + \sqrt n}$
have been assigned except one variable for some $i$,
$0 \leq i \leq \sqrt n - 1$, then we assign 0 or 1 to the variable so that
${\rm Parity}_{\sqrt n}(x_{\sqrt n \cdot i + 1}, x_{\sqrt n \cdot i + 2}, \ldots,
 x_{\sqrt n \cdot i + \sqrt n}) = 0$ and we do not consider the number of
eliminated gates.
By the modification, we can eliminate at least $3n - o(n)$ gates.
\end{proof}

\begin{proof}[Proof of Theorem~\ref{thrm:ndc_main}]
By Lemma~\ref{lem:ndc_ub} and Lemma~\ref{lem:ndc_lb}, the theorem holds.
\end{proof}

\section{Bounded space, nondeterminism, and large advice} \label{sec:l}

\subsection{Preliminaries} \label{subsec:l_pre}

Let $n$ be the input size.
{\sf L} is the class of decision problems solvable by a $O(\log n)$ space
Turing machine.
{\sf NL} is the nondeterministic variant of {\sf L}.
{\sf L/poly} is the class of decision problems solvable by a $O(\log n)$ space
Turing machine with polynomial size advice.
{\sf NL/poly} is the nondeterministic variant of {\sf L/poly}.
{\sf L/quasipoly} is the class of decision problems solvable by a $O(\log n)$
space Turing machine with quasipolynomial size advice.

\subsection{Proof of Theorem~\ref{thrm:l_conv}} \label{subsec:l_conv_p}

\begin{lemma} \label{lem:l_conv}
If any nondeterministic circuit of size $s + \lceil w/2 \rceil$ and width $w$
can be converted to a deterministic circuit of size $s'$ and width $w'$,
then any nondeterministic circuit of size $2s$ and width $w$ can be converted
to a deterministic circuit of size $O(2^ws')$ and width $w'+2$.
\end{lemma}

\begin{proof}
Let $C$ be a nondeterministic circuit of size $2s$ and width $w$.
We separate $C$ to two circuits at a layer such that each two circuits
has size at most $s + \lceil w/2 \rceil$.
Let $C_1$ and $C_2$ be the former circuit and the latter circuit, respectively.
The number of values (0 or 1) from $C_1$ is at most $w$.
For each at most $2^w$ combination of 0 and 1, we check whether both of
$C_1$ and $C_2$ are satisfied.
Natural construction of such circuit is enough to prove the lemma. % todo
\end{proof}

\begin{proof}[Proof of Theorem~\ref{thrm:l_conv}]
We apply Lemma~\ref{lem:l_conv} recursively.
\end{proof}

%\begin{lemma} \label{lem:l_conv}
%If any nondeterministic branching program of size $s$ can be converted
%to a Boolean circuit of size $s'$ and width $w$,
%then any nondeterministic branching program of size $2s$ can be converted
%to a Boolean circuit of size $2s^2s' + O(s^2)$ and width $w+2$.
%\end{lemma}
%\begin{proof}
%Let $G$ be a nondeterministic branching programs of size $2s$.
%Let $G_1$ and $G_2$ be the former $s$ nodes and the latter $s$ nodes,
%respectively, in arbitrary topological sorted order.
%Let $E_1$ be the edges between $G_1$ and $G_2$.
%The number of edges in $E_1$ is at most $s^2$.
%All paths from the start node to a sink node contain one edge in $E_1$.
%For each edge in $E_1$, we check the existence of a path from the start
%node to the 1 sink node.
%Natural construction of such circuit is enough to prove the lemma.
%%
%% More precisely, the construction of the circuit is as follows.
%% todo
%\end{proof}

\subsection{Proof of Theorem~\ref{thrm:l_main}} \label{subsec:l_main_p}

In this subsection, we prove Theorem~\ref{thrm:l_main}
by Theorem~\ref{thrm:l_conv}.

\begin{proof}[Proof of Theorem~\ref{thrm:l_main}]
Let $n$ be the size of the input.
We apply polynomial of $n$ to $s$ and $O(\log n)$ to $w$
in Theorem~\ref{thrm:l_conv}.
Then, we obtain that any nondeterministic circuit of
polynomial size of $n$ and width $O(\log n)$ can be converted
to a deterministic circuit of size $2^{O(\log^2 n)}$ and width $O(\log n)$.
Nondeterministic circuits of polynomial size and width $O(\log n)$ correspond
to {\sf NL/poly}. % todo
Deterministic circuits of size $2^{O(\log^2 n)}$ and width $O(\log n)$
correspond to {\sf L/quasipoly}. % todo
\end{proof}

\section{Concluding Remarks and Open Problems}

In this paper, we considered bounded width circuits and nondeterministic
circuits in three somewhat new directions.
Many open problems are raised after this paper.

In the first part of this paper, we proved the lower bounds for bounded width
circuits.
We remark a relation between bounded width circuits and bounded depth circuits.
\begin{theorem}
Any (deterministic) circuit of size $O(n)$ and depth $(\log n)$ can be
converted to a (deterministic) circuit of size $O(n^{1 + \epsilon})$ and
width $O(n / \log\log n)$.
\end{theorem}
\begin{proof}[Proof sketch]
Let $C$ be a circuit of size $O(n)$ and depth $(\log n)$.
It is known that we can find $O(n / \log\log n)$ edges in $C$ whose
removal yields a circuit of depth at most $\epsilon \log n$
(\cite{V76}, Section~14.4.3 of \cite{AB09}).
We construct a circuit which computes the value of $O(n / \log\log n)$ edges
one by one.
\end{proof}
Thus, improvement of Theorem~\ref{thrm:bwc_main} (or
Theorem~\ref{thrm:bwc_brs}) is also an approach to prove the lower bounds
for bounded depth circuits.

In the second part of this paper, we considered the power of nondeterministic
circuits. To prove the main theorem (Theorem~\ref{thrm:ndc_main}), we
introduced a simple proof strategy using nondeterministic selecting.
It remains open that we use the strategy and prove a similar or improved
result of Theorem~\ref{thrm:ndc_main} for $U_2$-circuits or other
Boolean circuits. % todo

In the third part of this paper, we proved that
$\mathsf{L/quasipoly} \supset \mathsf{NL}$ (Corollary~\ref{coro:l_main}).
It remains open whether this result can be improved to
$\mathsf{L/poly} \supset \mathsf{NL}$.
Another direction is revealing the power of large advice in {\sf L}.
In this paper, we proved that {\sf L} with quasipolynomial size advice
has nontrivial computational power.
It may be interesting that some relations between {\sf L} with advice
beyond polynomial size and other complexity classes ({\sf P},
{\sf PSPACE} and so on) are proved.

\bibliographystyle{plain}
\bibliography{circuit}

\end{document}